\documentclass[letterpaper,10pt,conference]{ieeeconf}
\IEEEoverridecommandlockouts
\usepackage{cite}
\usepackage{amsmath,amssymb,amsfonts}
\newtheorem{theorem}{Theorem}
\newtheorem{lemma}{Lemma}
\newtheorem{assumption}{Assumption}
\newtheorem{definition}{Definition}
\newtheorem{proposition}{Proposition}
\newtheorem{remark}{Remark}

\usepackage{algorithmic}
\usepackage{graphicx}
\usepackage{textcomp}
\usepackage{xcolor}
\usepackage[caption=false,font=footnotesize]{subfig}
\def\BibTeX{{\rm B\kern-.05em{\sc i\kern-.025em b}\kern-.08em
    T\kern-.1667em\lower.7ex\hbox{E}\kern-.125emX}}

\begin{document}

\title{Data-Driven Synthesis of Robust Positively Invariant Sets from Noisy Data\\
{\footnotesize}
\thanks{Chi Wang and David Angeli are with the Control and Power Group, Department of Electrical and Electronic Engineering, Imperial College London, London SW7 2AZ, U.K. (\{chi.wang23,d.angeli\}@imperial.ac.uk).
}
}

\author{Chi Wang, David Angeli
}
\maketitle

\begin{abstract}
This paper develops a method to construct robust positively invariant (RPI) tube sets from finite noisy input-state data of an unknown linear time-invariant (LTI) system, yielding tubes that can be directly embedded in tube-based robust data-driven predictive control. Data-consistency uncertainty sets are constructed under process/measurement noise with polytopic/ellipsoidal bounds. In the measurement-noise case, we provide a deterministic and data-consistent procedure to certify the induced residual bound from data. Based on these sets, a robustly stabilizing state-feedback gain is certified via a common quadratic contraction, which in turn enables constructive polyhedral/ellipsoidal RPI tube computation. Numerical examples quantify the conservatism induced by noisy data and the employed certification step.
\end{abstract}

\section{Introduction}
Robust positively invariant (RPI) sets are a fundamental tool in control and optimization, with
particular importance in predictive control and constraint handling~\cite{Blanchini1999}. In model-based design,
the computation and approximation of minimal RPI (mRPI) sets and maximal RPI (or invariant) sets have been
studied extensively~\cite{Rakovic2005MRPI,Trodden2016,KolmanovskyGilbert1998,GilbertTan2002}, and they play a central role in tube-based robust model predictive control
(RMPC)~\cite{Langson2004}. In particular, tube-based RMPC relies on an RPI set for the error dynamics to construct
a robust tube around the nominal trajectory and to synthesize an RPI terminal set, thereby enabling
robust constraint satisfaction and recursive feasibility under bounded disturbances~\cite{Mayne2005,RawlingsMayneDiehl2017}.

Motivated by the increasing availability of data and the cost of reliable modeling, data-driven
control has recently emerged as a powerful alternative to classical model-based approaches. The key
idea is to synthesize controllers directly from a finite set of input-state trajectories of an
unknown linear time-invariant (LTI) system. Starting from Willems' fundamental lemma~\cite{Willems2005}, behavioral descriptions have been leveraged to develop a range of data-driven analysis and control tools, with representative results on stabilization~\cite{DePersisTesi2020TAC} and state-feedback synthesis under noisy data~\cite{Berberich2020ACC,bisoffi2021trade}. These results indicate that invariant-set-based robust designs can, in principle, be performed directly from data.

In parallel, for data-enabled predictive control (DeePC)~\cite{coulson2019data} and related data-driven model predictive control (MPC)
formulations~\cite{Huang2019CDC}, which primarily establish open-loop performance, deriving rigorous closed-loop
guarantees in the presence of noise remains challenging. To bridge this gap, the robust data-driven MPC
frameworks proposed in~\cite{Berberich2021TAC,Berberich2020} provide some of the first theoretical guarantees;
however, the resulting bounds can be conservative. Other formulations can reduce conservatism at the expense of more demanding computations, e.g.,
infinite-horizon semidefinite programming (SDP) designs or repeated large linear matrix inequality (LMI) problems~\cite{xie2026,Hu2025TAC}. In
contrast, the tube-based paradigm---with its reliance on quadratic programs, tunable prediction
horizons, and a transparent link between disturbance bounds and conservatism---offers an attractive alternative
route for robust data-driven predictive control; however, it has not yet been fully exploited in a data-driven manner.

This paper takes a step in this direction by developing explicit and computable procedures to
construct RPI tube sets directly from noisy input-state trajectories. We consider both process-noise
and measurement-noise scenarios, with both polytopic and ellipsoidal noise
descriptions. A key difficulty in the measurement-noise case is that the induced regression residual
depends on the unknown system matrix, which necessitates careful data-consistency modeling. Building on these ingredients, we provide a unified framework that yields (i) data-consistency sets
under four noise scenarios and (ii) corresponding RPI tube sets that are directly usable within
tube-based data-driven predictive control.

Related work includes data-driven invariance and controller synthesis in nominal settings~\cite{BisoffiDePersisTesi2020IFAC}.
More recently,~\cite{BisoffiDePersisTesi2022} provides a sharp characterization of robust invariance directly from noisy
input-state data: it derives necessary and sufficient conditions for the existence of a linear state-feedback gain that
renders a prescribed polyhedral set robustly invariant, and shows that such a gain can be obtained by solving a linear program.
This perspective primarily targets feasibility certification for a given set, rather than an end-to-end construction of RPI tubes.
Complementarily,~\cite{MulagaletiBemporadZanon2021} considers joint data-driven synthesis of robust invariant sets and stabilizing controllers,
which can yield less conservative tubes than sequential designs. While powerful, these formulations typically require initialization and numerical tuning,
which are not directly data-checkable. In contrast, our approach provides a unified, offline-computable procedure for the four
noise settings considered in this paper, with tractable synthesis and certification steps for robust stabilizing feedback and for constructing certified polyhedral/ellipsoidal RPI tubes,
thereby enabling a transparent quantification of the conservatism induced by noisy data. The remainder is organized as follows. Section~\ref{sec:Preliminaries} introduces the preliminaries and problem setting; Section~\ref{sec:Data-consistent} develops the data-consistency descriptions and robust-feedback synthesis; Section~\ref{sec:RPI} constructs polyhedral and ellipsoidal RPI tubes; Section~\ref{sec:examples} evaluates their conservatism numerically. Section~\ref{sec:conclusions}
concludes the paper.
\section{Preliminaries}\label{sec:Preliminaries}
Section~\ref{sec:Notation} fixes notation, while Section~\ref{sec:problem} states the noisy LTI setting, offline data assumptions, and tube-based robustness notions used throughout the paper.

\subsection{Notation}\label{sec:Notation}
Let $\mathbb{N}_{\ge 0}:=\{0,1,2,\ldots\}$ and let $\mathbb{R}^n$ denote the $n$-dimensional real space.
For a symmetric matrix $A$, $A\succeq 0$ ($A\succ 0$) means that $A$ is positive semidefinite (positive definite), and
$\mathbb{S}_{\succeq 0}^n$ ($\mathbb{S}_{\succ 0}^n$) denotes the set of $n\times n$ symmetric positive semidefinite (positive definite) matrices.
The identity matrix is $I$, and $(\cdot)^\top$ denotes transpose.
The Euclidean norm is $\|\cdot\|_2$, and $\|A\|_2$ denotes the induced $2$-norm.
Moreover, $\lambda_{\max}(\cdot)$ denotes the largest eigenvalue of a symmetric matrix,
$\operatorname{rank}(\cdot)$ the matrix rank, and $(\cdot)^\dagger$ the Moore--Penrose pseudoinverse.
For sets $\mathcal S,\mathcal T\subset\mathbb{R}^n$, the Minkowski sum is
$\mathcal S\oplus\mathcal T:=\{\,s+t\mid s\in\mathcal S,\ t\in\mathcal T\,\}$, and for $\alpha\ge0$,
$\alpha\mathcal S:=\{\,\alpha s\mid s\in\mathcal S\,\}$.
For $A\in\mathbb{R}^{m\times n}$ and $\mathcal S\subset\mathbb{R}^n$, the linear image is
$A\mathcal S:=\{\,As\mid s\in\mathcal S\,\}\subset\mathbb{R}^m$.
Furthermore, $\operatorname{co}(\cdot)$ denotes the convex hull, $\operatorname{vert}(\mathcal S)$ the vertex set of a polytope,
and $\operatorname{int}(\mathcal S)$ the interior.

\subsection{Problem setup}\label{sec:problem}
We consider an unknown discrete-time LTI system with state $x_k\in\mathbb{R}^n$ and input
$u_k\in\mathbb{R}^m$. The matrices $A^\ast\in\mathbb{R}^{n\times n}$ and
$B^\ast\in\mathbb{R}^{n\times m}$ denote the (unknown) true system dynamics.
We study robustness with respect to either bounded process noise or
bounded measurement noise, leading to the following two scenarios.

\paragraph{Process noise}
The state evolves according to
\begin{equation}\label{eq:process_noise_model}
    x_{k+1} = A^\ast x_k + B^\ast u_k + w_k,
\end{equation}
where $w_k\in\mathbb{R}^n$ is the unknown-but-bounded process noise.

\paragraph{Measurement noise}
The state evolves noiselessly, while measurements are corrupted by noise:
\begin{subequations}\label{eq:measurement_noise_model}
\begin{align}
    x_{k+1} &= A^\ast x_k + B^\ast u_k, \label{eq:lti_system}\\
    \hat{x}_k &= x_k + v_k, \label{eq:meas_eq}
\end{align}
\end{subequations}
where $\hat{x}_k\in\mathbb{R}^n$ denotes the measured full state and $v_k\in\mathbb{R}^n$
is the unknown-but-bounded measurement noise.

In both scenarios, the noise is modeled via a known compact convex set (either polytopic
or ellipsoidal), which is standard in invariant-set-based robust control and synthesis~\cite{Langson2004}.

\begin{assumption}\label{as:bounded_noise}
For all $k\in\mathbb{N}_{\ge0}$, the noise satisfies $w_k\in\mathcal W$ and
$v_k\in\mathcal V$, where $\mathcal W,\mathcal V\subset\mathbb{R}^n$ are known, nonempty, compact, convex sets
satisfying $0\in\mathcal W$ and $0\in\operatorname{int}(\mathcal V)$. In particular, we consider two commonly used parameterizations:
\emph{(i) Polytopic sets:}
\begin{subequations}\label{eq:noise_poly_sets}
\begin{align}
\mathcal W_{\mathrm P} &:= \{\, w \in \mathbb{R}^{n} \mid G_w w \le h_w \,\}, \label{eq:W_polytope}\\
\mathcal V_{\mathrm P} &:= \{\, v \in \mathbb{R}^{n} \mid G_v v \le h_v \,\}, \label{eq:V_polytope}
\end{align}
\end{subequations}
where $G_w\in\mathbb{R}^{q_w\times n},\, h_w\in\mathbb{R}^{q_w}$ and
$G_v\in\mathbb{R}^{q_v\times n},\, h_v\in\mathbb{R}^{q_v}$ are known.

\emph{(ii) Ellipsoidal sets:}
\begin{subequations}\label{eq:noise_ell_sets}
\begin{align}
\mathcal W_{\mathrm E} &:= \{\, w \in \mathbb{R}^{n} \mid w^\top Q_w^{-1} w \le 1 \,\}, \label{eq:W_ellipsoid}\\
\mathcal V_{\mathrm E} &:= \{\, v \in \mathbb{R}^{n} \mid v^\top Q_v^{-1} v \le 1 \,\}, \label{eq:V_ellipsoid}
\end{align}
\end{subequations}
where $Q_w,Q_v\in\mathbb{S}_{\succ0}^{n}$ are known.
\end{assumption}

Our goal is to construct an RPI set directly from offline collected noisy data
to enable data-driven predictive control, with particular emphasis on the terminal tube set and the
initial tube set, ensuring recursive feasibility. To this end, we assume access to a single offline trajectory
collected under one of the following noise scenarios.

\begin{assumption}\label{as:available-trajectory}
A trajectory $\big(u^d_{[0{:}T{-}1]},\,x^d_{[0{:}T{-}1]}\big)$ of length $T$
generated by system~\eqref{eq:process_noise_model}, or a trajectory
$\big(u^d_{[0{:}T{-}1]},\,\hat{x}^d_{[0{:}T{-}1]}\big)$ of length $T$
generated by system~\eqref{eq:measurement_noise_model}, is available.
\end{assumption}

From the available trajectory, we form the standard data matrices
\begin{align}
X_0&=[x^d_0\ \cdots\ x^d_{T-2}],\qquad
X_1=[x^d_1\ \cdots\ x^d_{T-1}],\\
U_0&=[u^d_0\ \cdots\ u^d_{T-2}],
\end{align}
and, in the measurement-noise case, their measured counterparts
\[
\hat X_0=[\hat x^d_0\ \cdots\ \hat x^d_{T-2}],\qquad
\hat X_1=[\hat x^d_1\ \cdots\ \hat x^d_{T-1}].
\]
These data serve as the basic ingredients for constructing Hankel matrices and, in turn, data-driven
predictors. In the presence of noise, however, the Hankel matrices built from these noisy data generally fail to
span the exact trajectory space. The resulting prediction mismatch can be interpreted as an induced
multiplicative uncertainty in the data-driven predictor~\cite{Berberich2021TAC}. To facilitate tube-based design, we therefore adopt a
unified additive representation that captures all such effects.

Specifically, for any robustly stabilizing gain $K$ synthesized from the data-consistency set (see Section~\ref{sec:Data-consistent}), the tube-error dynamics admit the form
\begin{equation}\label{eq:state_map_TAC_refined}
e_{k+1}=(A^\ast+B^\ast K)\,e_k+d_k,\qquad d_k\in\mathcal D,
\end{equation}
where $d_k$ aggregates all uncertainty sources, including process/measurement noise and the data-induced prediction mismatch.
To characterize tube sets that are invariant under~\eqref{eq:state_map_TAC_refined}, we recall the standard notion of robust
positive invariance.

\begin{definition}\label{def:RPI}
A set $\mathcal E\subset\mathbb R^n$ is said to be RPI for the dynamics~\eqref{eq:state_map_TAC_refined} with residual set $\mathcal D$ if
\[
e_k\in\mathcal E \ \Longrightarrow\ e_{k+1}\in\mathcal E,\qquad \forall\,d_k\in\mathcal D.
\]
Equivalently, $\mathcal E$ is RPI if and only if it satisfies the set inclusion
\begin{equation}\label{eq:RPI_condition_TAC_refined}
(A^\ast+B^\ast K)\,\mathcal E \ \oplus\ \mathcal D \ \subseteq\ \mathcal E.
\end{equation}
\end{definition}

\section{Data-consistency sets and robust stabilizing controller}\label{sec:Data-consistent}
We construct data-consistency sets for polytopic and ellipsoidal process/measurement noise, each containing the unknown true dynamics $(A^\ast,B^\ast)$, and then formulate robust stabilizing controller synthesis over these sets.
\subsection{Data-consistent set}
We first treat process noise, for which the one-step relation gives an immediate data-consistency description.

\emph{(i) Polytopic process noise.}
Based on the bound $\mathcal W_{\mathrm P}$ in~\eqref{eq:W_polytope}, define the one-step data-consistency set as
\begin{equation}\label{eq:Ci_poly_TAC}
\mathcal C_i^{\mathcal W_{\mathrm P}}
:=
\Bigl\{ (A,B)\Big|
\exists w_i\in \mathcal W_{\mathrm P}:\
x^d_{i+1}=A x^d_i+B u^d_i+w_i
\Bigr\}.
\end{equation}
The set of all matrices consistent with the data is then
\begin{equation}\label{eq:I_set_TAC_final_WP}
\mathcal I^{\mathcal W_{\mathrm P}}
:=
\bigcap_{i=0}^{T-2}\mathcal C_i^{\mathcal W_{\mathrm P}}.
\end{equation}
\emph{(ii) Ellipsoidal process noise.}
Similarly, for the ellipsoidal bound $\mathcal W_{\mathrm E}$ in~\eqref{eq:W_ellipsoid}, define
\begin{equation}\label{eq:Ci_set_TAC_final_WE}
\mathcal C_i^{\mathcal W_{\mathrm E}}
:=
\Bigl\{(A,B) \Big|
\exists w_i\in \mathcal W_{\mathrm E}:\;
x^d_{i+1}=A x^d_i+B u^d_i+w_i
\Bigr\}.
\end{equation}
The corresponding data-consistency set is
\begin{equation}\label{eq:I_set_TAC_final_WE}
\mathcal I^{\mathcal W_{\mathrm E}}
:=
\bigcap_{i=0}^{T-2} \mathcal C_i^{\mathcal W_{\mathrm E}}.
\end{equation} 

We now consider the measurement-noise model~\eqref{eq:measurement_noise_model}, where only
$\hat x_k=x_k+v_k$ is available and the measurement noise satisfies $v_k\in\mathcal V$.
In particular, the measured state obeys the regression
\begin{equation}\label{eq:regression_meas_unified}
\hat x_{k+1}=A^\ast \hat x_k + B^\ast u_k + \eta_k,
\qquad
\eta_k:= v_{k+1}-A^\ast v_k.
\end{equation}
Since $\eta_k$ depends on $A^\ast$, it is generally not constrained to lie in the original noise set.
To obtain a tractable data-consistency description, we bound $\eta_k$ via the gauge induced by $\mathcal V$.
Define the Minkowski functional (gauge) and its induced matrix gain as
\[
\|z\|_{\mathcal V}
:=
\inf\{\,\alpha\ge0 \mid z\in\alpha\mathcal V\,\},
\qquad
\|A\|_{\mathcal V}
:=
\sup_{z\neq0}\frac{\|Az\|_{\mathcal V}}{\|z\|_{\mathcal V}}.
\]
For the measurement residual, define the directed induced gain
$\mu_{\mathcal V}(A)$ by
\[
\mu_{\mathcal V}(A)
:=\|-A\|_{\mathcal V}
=\min\{\,\gamma\ge0\mid -A\mathcal V\subseteq\gamma\mathcal V\,\}.
\]
If $\mathcal V$ is centrally symmetric, then
$\mu_{\mathcal V}(A)=\|A\|_{\mathcal V}$.

\begin{lemma}\label{lem:gauge_inflation_meas}
For any $v_k,v_{k+1}\in\mathcal V$ and any matrix $A$,
$v_{k+1}-Av_k\in(1+\mu_{\mathcal V}(A))\mathcal V$.
Consequently, for any $\gamma\ge\mu_{\mathcal V}(A)$,
\begin{equation}\label{eq:D_gamma_general}
\mathcal D(\gamma):=(1+\gamma)\mathcal V,
\qquad
v_{k+1}-Av_k\in\mathcal D(\gamma).
\end{equation}
\end{lemma}
\begin{proof}
By definition, $-A\mathcal V\subseteq\mu_{\mathcal V}(A)\mathcal V$. Hence
\[
v_{k+1}-Av_k
\in\mathcal V\oplus\mu_{\mathcal V}(A)\mathcal V
=\bigl(1+\mu_{\mathcal V}(A)\bigr)\mathcal V,
\]
where the equality follows from convexity of $\mathcal V$. Since scalar dilations of
$\mathcal V$ are nested,~\eqref{eq:D_gamma_general} follows for every
$\gamma\ge\mu_{\mathcal V}(A)$.
\end{proof}

Specializing~\eqref{eq:regression_meas_unified} to the offline dataset
$\big(u^d_{[0{:}T{-}1]},\hat x^d_{[0{:}T{-}1]}\big)$ from Assumption~\ref{as:available-trajectory},
Lemma~\ref{lem:gauge_inflation_meas} implies that, for any
$\gamma\ge\mu_{\mathcal V}(A^\ast)$, the measured data generated by the true pair satisfy
\begin{equation}\label{eq:one_step_meas_unified}
\hat x^d_{i+1}
=
A^\ast\hat x^d_i + B^\ast u^d_i + \eta_i^\ast,
\eta_i^\ast\in\mathcal D(\gamma), i=0,\dots,T-2.
\end{equation}
This motivates the one-step consistency set
\begin{equation}\label{eq:Ci_meas_template}
\mathcal C_i^{\mathcal V}(\gamma)
:=
\Bigl\{(A,B) \Big| \exists \eta_i\in\mathcal D(\gamma) :\
\hat x^d_{i+1}=A\hat x^d_i+B u^d_i+\eta_i
\Bigr\},
\end{equation}
and the corresponding data-consistency set
\begin{equation}\label{eq:I_meas_template}
\mathcal I^{\mathcal V}(\gamma)
:=
\bigcap_{i=0}^{T-2}\mathcal C_i^{\mathcal V}(\gamma).
\end{equation}

To obtain nonempty and computable data-consistency sets, we select inflation parameters
that are certified upper bounds on the residual-induced gains of the true matrix,
\begin{equation}\label{eq:gamma_star_req}
\gamma_{\mathrm P}^\ast \ \ge\ \mu_{\mathcal V_{\mathrm P}}(A^\ast),
\qquad
\gamma_{\mathrm E}^\ast \ \ge\ \mu_{\mathcal V_{\mathrm E}}(A^\ast)
=\|A^\ast\|_{\mathcal V_{\mathrm E}}.
\end{equation}
Appendix~\ref{app:gamma-tight} provides offline-checkable constructions for selecting these bounds.

Fixing $\gamma\in\{\gamma_{\mathrm P}^\ast,\gamma_{\mathrm E}^\ast\}$ removes the homogeneity in the description of
$\mathcal C_i^{\mathcal V}(\gamma)$ and $\mathcal I^{\mathcal V}(\gamma)$ and yields tractable sets
$\mathcal I^{\mathcal V_{\mathrm P},\ast}$ and $\mathcal I^{\mathcal V_{\mathrm E},\ast}$.

\emph{(i) Polytopic measurement noise.}
For $\mathcal V=\mathcal V_{\mathrm P}$ in~\eqref{eq:V_polytope}, the gauge is well-defined and
$\alpha\mathcal V_{\mathrm P}=\{z\mid G_v z\le \alpha h_v\}$. Hence
\[
\mathcal D_{\mathrm P}(\gamma)
:=
(1+\gamma)\mathcal V_{\mathrm P}
=
\{\,\eta\mid G_v\eta\le (1+\gamma)h_v\,\}.
\]
Applying~\eqref{eq:Ci_meas_template} with $\mathcal D=\mathcal D_{\mathrm P}(\gamma)$ yields
\begin{equation}\label{eq:Ci_VP_gamma}
\mathcal C_i^{\mathcal V_{\mathrm P}}(\gamma)
\!:=
\Bigl\{(A,B) \Big|
\exists \eta_i\in\mathcal D_{\mathrm P}(\gamma):
\hat x^d_{i+1}=A\hat x^d_i+B u^d_i+\eta_i
\Bigr\},
\end{equation}
and $\mathcal I^{\mathcal V_{\mathrm P}}(\gamma):=\bigcap_{i=0}^{T-2}\mathcal C_i^{\mathcal V_{\mathrm P}}(\gamma)$.
Fixing $\gamma=\gamma_{\mathrm P}^\ast$ in~\eqref{eq:Ci_VP_gamma} defines the one-step sets
$\mathcal C_i^{\mathcal V_{\mathrm P},\ast}$ and the polytope
$\mathcal I^{\mathcal V_{\mathrm P},\ast}$ as
\begin{equation}\label{eq:I_VP_star}
\mathcal C_i^{\mathcal V_{\mathrm P},\ast}
:=
\mathcal C_i^{\mathcal V_{\mathrm P}}(\gamma_{\mathrm P}^\ast),
\qquad
\mathcal I^{\mathcal V_{\mathrm P},\ast}
:=
\bigcap_{i=0}^{T-2}\mathcal C_i^{\mathcal V_{\mathrm P},\ast}.
\end{equation}

\emph{(ii) Ellipsoidal measurement noise.}
For $\mathcal V=\mathcal V_{\mathrm E}$ in~\eqref{eq:V_ellipsoid}, the gauge is
$\|z\|_{\mathcal V_{\mathrm E}}=\sqrt{z^\top Q_v^{-1}z}$, and central symmetry gives
$\mu_{\mathcal V_{\mathrm E}}(A)=\|A\|_{\mathcal V_{\mathrm E}}$. Hence, for any
$\gamma\ge\|A\|_{\mathcal V_{\mathrm E}}$, Lemma~\ref{lem:gauge_inflation_meas} gives
$\eta_k\in(1+\gamma)\mathcal V_{\mathrm E}$. Accordingly,
\begin{equation}\label{eq:D_gamma_VE}
\mathcal D_{\mathrm E}(\gamma)
:=
(1+\gamma)\mathcal V_{\mathrm E}
=
\bigl\{\,\eta\in\mathbb R^n \ \big|\ \eta^\top Q_v^{-1}\eta\le (1+\gamma)^2\,\bigr\}.
\end{equation}
Applying~\eqref{eq:Ci_meas_template} with $\mathcal D=\mathcal D_{\mathrm E}(\gamma)$ yields
\begin{equation}\label{eq:Ci_VE_gamma}
\mathcal C_i^{\mathcal V_{\mathrm E}}(\gamma)
\!:=
\Bigl\{(A,B) \Big|
\exists \eta_i\in\mathcal D_{\mathrm E}(\gamma) :
\hat x^d_{i+1}=A\hat x^d_i+B u^d_i+\eta_i
\Bigr\},
\end{equation}
and $\mathcal I^{\mathcal V_{\mathrm E}}(\gamma):=\bigcap_{i=0}^{T-2}\mathcal C_i^{\mathcal V_{\mathrm E}}(\gamma)$.
Fixing $\gamma=\gamma_{\mathrm E}^\ast$ in~\eqref{eq:Ci_VE_gamma} defines the one-step sets
$\mathcal C_i^{\mathcal V_{\mathrm E},\ast}$ and the spectrahedron $\mathcal I^{\mathcal V_{\mathrm E},\ast}$ as
\begin{equation}\label{eq:I_VE_star}
\mathcal C_i^{\mathcal V_{\mathrm E},\ast}
:=
\mathcal C_i^{\mathcal V_{\mathrm E}}(\gamma_{\mathrm E}^\ast),
\qquad
\mathcal I^{\mathcal V_{\mathrm E},\ast}
:=
\bigcap_{i=0}^{T-2}\mathcal C_i^{\mathcal V_{\mathrm E},\ast}.
\end{equation}

To ensure boundedness of the resulting consistency sets, we impose the following standard rank condition on collected data.
\begin{assumption}\label{as:fullrowrank}
Let
\[
Z_x:=
\begin{cases}
X_0, & \text{for process noise},\\
\hat X_0, & \text{for measurement noise},
\end{cases}
\qquad
Z:=\begin{bmatrix}Z_x\\ U_0\end{bmatrix}.
\]
The matrix $Z\in\mathbb R^{(n+m)\times(T-1)}$ has full row rank, i.e.,
$\operatorname{rank}(Z)=n+m$.
\end{assumption}
In the noise-free case, this property is guaranteed when the input
$u^d_{[0{:}T{-}1]}$ is persistently exciting of order $n+1$ and the pair $(A^\ast,B^\ast)$
is controllable~\cite{Willems2005}. With noise, we assume that the collected trajectory is
nondegenerate so that the rank condition is preserved.
\begin{lemma}\label{lem:I_compact}
Suppose Assumption~\ref{as:bounded_noise},~\ref{as:available-trajectory}
and~\ref{as:fullrowrank} hold and, for measurement noise, the bounds satisfy~\eqref{eq:gamma_star_req}. Then
$\mathcal I^{\mathcal W_{\mathrm P}}$ and $\mathcal I^{\mathcal V_{\mathrm P},\ast}$ are nonempty compact
polytopes, while $\mathcal I^{\mathcal W_{\mathrm E}}$ and $\mathcal I^{\mathcal V_{\mathrm E},\ast}$ are
nonempty compact spectrahedra.
\end{lemma}

\begin{proof}
Closedness and convexity follow from the defining affine or quadratic constraints. Nonemptiness follows from
$w_i\in\mathcal W$ in the process-noise case and, in the measurement-noise case, from the certified bounds
in Appendix~\ref{app:gamma-tight}.

For boundedness, let $X_1^{\mathrm{dat}}:=X_1$ for process noise and
$X_1^{\mathrm{dat}}:=\hat X_1$ for measurement noise. Every feasible pair satisfies
\begin{equation}\label{eq:stack_consistency}
X_1^{\mathrm{dat}}=[A\ B]Z+E_0,
\qquad E_0=[e_0\ \cdots\ e_{T-2}],
\end{equation}
where the columns of $E_0$ belong, respectively, to
$\mathcal W_{\mathrm P}$, $\mathcal W_{\mathrm E}$,
$(1+\gamma_{\mathrm P}^\ast)\mathcal V_{\mathrm P}$, or
$(1+\gamma_{\mathrm E}^\ast)\mathcal V_{\mathrm E}$. These sets are compact, so $\|E_0\|$ is uniformly
bounded. Since $ZZ^\dagger=I_{n+m}$,~\eqref{eq:stack_consistency} gives
\[
\begin{bmatrix}A & B\end{bmatrix}
=
(X_1^{\mathrm{dat}}-E_0)Z^\dagger,
\]
and therefore
\[
\Bigl\|\begin{bmatrix}A & B\end{bmatrix}\Bigr\|
\le
(\|X_1^{\mathrm{dat}}\|+\|E_0\|)\,\|Z^\dagger\|.
\]
Thus all four sets are bounded and hence compact. The polytopic cases are polyhedra and the ellipsoidal
cases admit LMI representations, completing the proof.
\end{proof}
\subsection{Robust stabilizing controller}
We next design a robust stabilizing tube feedback gain $K$ for the data-consistency sets
$\mathcal I\in\{\mathcal I^{\mathcal W_{\mathrm P}},\,\mathcal I^{\mathcal V_{\mathrm P},\ast},\,
\mathcal I^{\mathcal W_{\mathrm E}},\,\mathcal I^{\mathcal V_{\mathrm E},\ast}\}$.
Specifically, we seek a common quadratic Lyapunov function with an explicit contraction margin:
\begin{equation}\label{eq:stability_lmi}
\begin{aligned}
\text{find}\quad & P \succ 0,\; K,\; \beta>0\\
\text{s.t.}\quad &
(A{+}BK)\,P\,(A{+}BK)^\top - P \prec -\beta I,
\quad \forall (A,B)\in\mathcal I.
\end{aligned}
\end{equation}
The margin $-\beta I$ enforces a uniform contraction rate, which is a key ingredient for the existence and computability
of an RPI tube set induced by the data-consistent uncertainty.

To obtain tractable conditions, we apply the standard change of variables $Y:=KP$.
By the Schur complement,~\eqref{eq:stability_lmi} is implied by
\begin{equation}\label{eq:decay_schur}
\begin{bmatrix}
P-\beta I & AP+BY\\
(AP+BY)^\top & P
\end{bmatrix}
\succeq 0,\qquad \forall (A,B)\in\mathcal I.
\end{equation}
We next derive computable conditions for~\eqref{eq:decay_schur}. The resulting formulations depend on the
geometry of $\mathcal I$: for ellipsoidal sets we use an S-procedure relaxation, whereas
for polytopic sets we obtain an exact vertex-based SDP.

\emph{(i) Ellipsoidal data-consistency sets.}
For $\mathcal I\in\{\mathcal I^{\mathcal W_{\mathrm E}},\mathcal I^{\mathcal V_{\mathrm E},\ast}\}$,
a computationally tractable sufficient condition for~\eqref{eq:decay_schur} can be obtained by the
lossy matrix S-procedure; see~\cite[Prop.~1]{bisoffi2021trade}.

\begin{proposition}\label{prop:K-stab-ellip}
Suppose there exist $P=P^\top\succ0$, $Y\in\mathbb{R}^{m\times n}$, $\beta>0$, and multipliers
$\tau_0,\ldots,\tau_{T-2}\ge0$ such that
\begin{equation}\label{eq:ellip_stab_sproc}
\begin{aligned}
&\begin{bmatrix}
P-\beta I & 0 & 0 & 0\\
0 & -P & -Y^\top & 0\\
0 & -Y & 0 & Y\\
0 & 0 & Y^\top & P
\end{bmatrix}
\;-\;
\sum_{i=0}^{T-2}\tau_i\,\Xi_i(Q)
\ \succeq\ 0 ,
\end{aligned}
\end{equation}
where
\begin{equation}\label{eq:Xi_i_def}
\Xi_i(Q)
:=
\tilde M_i
\begin{bmatrix}
Q & 0\\
0 & -1
\end{bmatrix}
\tilde M_i^\top,
\qquad
\tilde M_i
:=
\begin{bmatrix}
I & z^d_{i+1}\\
0 & -z^d_i\\
0 & -u^d_i\\
0 & 0
\end{bmatrix}.
\end{equation}
Then, with $K:=YP^{-1}$,
\[
(A{+}BK)\,P\,(A{+}BK)^\top - P + \beta I \preceq 0,
\quad \forall (A,B)\in\mathcal I.
\]
\end{proposition}

To match the two ellipsoidal data-consistency sets, we instantiate~\eqref{eq:Xi_i_def} as follows:
for $\mathcal I^{\mathcal W_{\mathrm E}}$, take $z_i^d:=x_i^d$ and $Q:=Q_w$;
for $\mathcal I^{\mathcal V_{\mathrm E},\ast}$, take $z_i^d:=\hat x_i^d$ and
$Q:=(1+\gamma_{\mathrm E}^\ast)^2 Q_v$.

\emph{(ii) Polytopic data-consistency sets.}
For $\mathcal I\in\{\mathcal I^{\mathcal W_{\mathrm P}},\,\mathcal I^{\mathcal V_{\mathrm P},\ast}\}$,
Condition~\eqref{eq:decay_schur} admits a lossless vertex characterization. Since $\mathcal I$ is a compact polytope, its
vertex set $\operatorname{vert}(\mathcal I)=\{(A^{(j)},B^{(j)})\}_{j=1}^{N_v}$ is finite, and it suffices to enforce~\eqref{eq:decay_schur}
at these vertices, as formalized next.

\begin{proposition}\label{prop:K-stab-poly}
Let $\mathcal I$ be a compact polytope with vertex set
$\operatorname{vert}(\mathcal I)=\{(A^{(j)},B^{(j)})\}_{j=1}^{N_v}$.
If there exist $P=P^\top\succ0$, $Y\in\mathbb{R}^{m\times n}$, and $\beta>0$ such that
\begin{equation}\label{eq:poly_decay_SDP}
\begin{aligned}
&\begin{bmatrix}
P-\beta I & S^{(j)}\\
\bigl(S^{(j)}\bigr)^\top & P
\end{bmatrix}\succeq 0,\qquad \forall\, j=1,\dots,N_v,\\
& S^{(j)} := A^{(j)}P + B^{(j)}Y ,
\end{aligned}
\end{equation}
then, with $K:=YP^{-1}$,
\begin{equation}\label{eq:robust_Lyap_poly_result}
(A{+}BK)\,P\,(A{+}BK)^\top - P + \beta I \preceq 0,
\quad \forall (A,B)\in\mathcal I.
\end{equation}
Moreover,~\eqref{eq:poly_decay_SDP} is lossless in the sense that~\eqref{eq:decay_schur} holds for all
$(A,B)\in\mathcal I$ if and only if it holds for all vertices $(A^{(j)},B^{(j)})$.
\end{proposition}

\begin{proof}
For fixed $(P,Y,\beta)$, the matrix in~\eqref{eq:decay_schur} depends on $(A,B)$ only through affine terms (e.g., $AP+BY$).
Since $\mathbb S_{\succeq 0}$ is a convex cone,~\eqref{eq:decay_schur} is convex in $(A,B)$. Hence, if~\eqref{eq:poly_decay_SDP} holds at all vertices, then for any
$(A,B)=\sum_{j=1}^{N_v}\lambda_j(A^{(j)},B^{(j)})$ with $\lambda_j\ge0$ and $\sum_j\lambda_j=1$,
\[
\begin{aligned}
\begin{bmatrix}
P-\beta I & AP+BY\\
(AP+BY)^\top & P
\end{bmatrix}
=\\
\sum_{j=1}^{N_v}\lambda_j
\begin{bmatrix}
P-\beta I & A^{(j)}P + B^{(j)}Y\\
\bigl(A^{(j)}P + B^{(j)}Y\bigr)^\top & P
\end{bmatrix}
\succeq 0,
\end{aligned}
\]
which proves~\eqref{eq:decay_schur} for all $(A,B)\in\mathcal I$. Substituting $Y=KP$ yields~\eqref{eq:robust_Lyap_poly_result}.
\end{proof}

\begin{remark}\label{rem:conserv_beta}
In the polytopic case, the vertex LMIs~\eqref{eq:poly_decay_SDP} are lossless and introduce no additional
conservatism. In contrast, the lossy S-procedure in~\eqref{eq:ellip_stab_sproc} yields a tractable sufficient
certificate. This relaxation may certify a smaller contraction margin (or enforce a more restrictive feedback gain),
thereby leading to a larger computed RPI tube (outer approximation).

To strengthen the certified closed-loop contraction---which typically yields a smaller mRPI set for the error dynamics and facilitates the construction of a larger feasible terminal set---the feasibility SDPs in~\eqref{eq:ellip_stab_sproc} and~\eqref{eq:poly_decay_SDP} can be lifted to optimization problems by maximizing the margin $\beta$ (subject to a suitable normalization to remove the inherent homogeneity). In addition, following~\cite{ChenAllgower1998}, the terminal set is typically scaled to ensure
terminal input admissibility.
\end{remark}

\section{Robust positively invariant sets}\label{sec:RPI}
We construct RPI tubes for the four data-consistency descriptions $\mathcal I\in\{\mathcal I^{\mathcal W_{\mathrm P}},\mathcal I^{\mathcal W_{\mathrm E}},\mathcal I^{\mathcal V_{\mathrm P},\ast},\mathcal I^{\mathcal V_{\mathrm E},\ast}\}$. Let $K$ be any robust stabilizing gain returned by
Propositions~\ref{prop:K-stab-ellip} or~\ref{prop:K-stab-poly}, so that there exist $P=P^\top\succ0$ and
$\beta>0$ satisfying the common quadratic contraction
\begin{equation}\label{eq:common_quad_contr}
A_K P A_K^\top - P \preceq -\beta I,
\qquad \forall\,A_K\in\mathcal A_K,
\end{equation}
where the closed-loop uncertainty set induced by $\mathcal I$ is
\begin{equation}\label{eq:A_K_def_unified}
\mathcal A_K
:= \bigl\{\,A{+}BK \ \big|\ (A,B)\in\mathcal I\,\bigr\}.
\end{equation}
The tube error then satisfies the additive uncertainty model
\begin{equation}\label{eq:err_dyn_unified}
e^+ \in \mathcal A_K e + \mathcal D,
\end{equation}
with the residual set $\mathcal D$ chosen consistently with the data model as
\begin{equation}\label{eq:Omega_cases}
\begin{aligned}
(\mathcal I,\mathcal D)\in\Bigl\{
&(\mathcal I^{\mathcal W_{\mathrm P}},\ \mathcal W_{\mathrm P}),\
(\mathcal I^{\mathcal W_{\mathrm E}},\ \mathcal W_{\mathrm E}),\\
&(\mathcal I^{\mathcal V_{\mathrm P},\ast},\ (1+\gamma_{\mathrm P}^\ast)\mathcal V_{\mathrm P}),\\
&(\mathcal I^{\mathcal V_{\mathrm E},\ast},\ (1+\gamma_{\mathrm E}^\ast)\mathcal V_{\mathrm E})
\Bigr\}.
\end{aligned}
\end{equation}

\emph{(i) Polytopic RPI sets.}
When $\mathcal I\in\{\mathcal I^{\mathcal W_{\mathrm P}},\,\mathcal I^{\mathcal V_{\mathrm P},\ast}\}$,
both $\mathcal A_K$ and $\mathcal D$ in~\eqref{eq:err_dyn_unified} are polytopic. In particular, by linearity of
$(A,B)\mapsto A{+}BK$, the closed-loop uncertainty set $\mathcal A_K$ is a polytope in the matrix space and thus admits
a finite vertex representation
$\operatorname{vert}(\mathcal A_K)=\{A_K^{(j)}\}_{j=1}^{N_K}$.
Define the set-valued map
\begin{equation}\label{eq:Phi_def_poly_replace}
\Phi(\mathcal S)
:=
\operatorname{co}\!\bigl(\mathcal A_K\, \mathcal S\bigr)\ \oplus\ \mathcal D.
\end{equation}
Since $\mathcal A_K$ is polytopic,
\[
\operatorname{co}\!\bigl(\mathcal A_K \mathcal S\bigr)
:=
\operatorname{co}\!\Bigl(\bigcup_{j=1}^{N_K} A_K^{(j)}\mathcal S\Bigr),
\]
which is a polytope as the convex hull of finitely many polytopes. Starting from any bounded seed
$\mathcal S_0\subset\mathbb R^n$, we generate the sequence
\begin{equation}\label{eq:Phi_iter_poly_replace}
\mathcal S_{t+1}=\Phi(\mathcal S_t),\qquad t=0,1,2,\ldots,
\end{equation}
which is well-defined and remains polytopic for all $t$.
A valid default initialization is $\mathcal S_0=\{0\}$, since $0\in\mathcal D$ implies
$\mathcal S_0\subseteq\Phi(\mathcal S_0)$.

To certify a polyhedral tube under $\varepsilon$-stopping, we measure set increments with respect to a compact,
convex, centrally symmetric polytope $\Omega\subset\mathbb R^n$ with $0\in\operatorname{int}(\Omega)$.
Let $\|\cdot\|_\Omega$ denote its Minkowski functional and define the induced operator norm
$\|A\|_\Omega:=\sup_{x\neq 0}\|Ax\|_\Omega/\|x\|_\Omega$. The worst-case induced gain over the closed-loop
uncertainty set $\mathcal A_K$ is
\begin{equation}\label{eq:c_def_poly_replace}
c_\Omega:=\sup_{A_K\in\mathcal A_K}\|A_K\|_\Omega
=\max_{j=1,\ldots,N_K}\|A_K^{(j)}\|_\Omega,
\end{equation}
where the vertex maximization follows since $\mathcal A_K$ is polytopic and $A\mapsto\|A\|_\Omega$ is convex. Let $\mathbb B_{P^{-1}}:=\{e\in\mathbb R^n:\|e\|_{P^{-1}}\le1\}$, where $\|e\|_{P^{-1}}:=\sqrt{e^\top P^{-1} e}$, and define the induced gain
$\|A\|_{P^{-1}}:=\sup_{e\neq0}\|Ae\|_{P^{-1}}/\|e\|_{P^{-1}}$. From~\eqref{eq:common_quad_contr}, it holds that
\begin{equation}\label{eq:cP_def_poly_replace}
c_{P^{-1}}:=\sup_{A_K\in\mathcal A_K}\|A_K\|_{P^{-1}}
\le \sqrt{1-\beta/\lambda_{\max}(P)}<1.
\end{equation}
We select $\Omega$ as a polyhedral outer approximation of
$\mathbb B_{P^{-1}}$.

\begin{lemma}\label{lem:cOmega_bound}
Let $\Omega$ satisfy $\mathbb B_{P^{-1}}\subseteq \Omega \subseteq \kappa\,\mathbb B_{P^{-1}}$ for some $\kappa\ge 1$.
Then $\|e\|_\Omega \le \|e\|_{P^{-1}} \le \kappa\|e\|_\Omega$ for all $e$, and thus
$\|A\|_\Omega \le \kappa\|A\|_{P^{-1}}$ for all $A$. Consequently, $c_\Omega \le \kappa c_{P^{-1}}$.
\end{lemma}

\begin{remark}\label{rem:kappa_computable}
One may take $\kappa=\max_{\omega\in\operatorname{vert}(\Omega)}\|\omega\|_{P^{-1}}$.
In particular, $\kappa<1/c_{P^{-1}}$ implies $c_\Omega<1$.
\end{remark}

The following theorem summarizes a standard set-iteration construction of a \emph{polyhedral} certified RPI tube via
$\varepsilon$-stopping.
\begin{theorem}\label{thm:RPI_poly}
Let $\mathcal A_K$ and $\mathcal D$ be the polytopes induced by
$\mathcal I\in\{\mathcal I^{\mathcal W_{\mathrm P}},\,\mathcal I^{\mathcal V_{\mathrm P},\ast}\}$, and let $\Phi$ and
$\{\mathcal S_t\}_{t\ge0}$ be defined by~\eqref{eq:Phi_def_poly_replace}--\eqref{eq:Phi_iter_poly_replace}.
Let $\Omega$ be chosen as in Remark~\ref{rem:kappa_computable}, so that $c_\Omega<1$. Then:
\emph{(i) Monotonicity and one-step boundedness:} if $\mathcal S\subseteq\mathcal T$, then
$\Phi(\mathcal S)\subseteq\Phi(\mathcal T)$. Moreover, if $\mathcal S$ is bounded, then so is $\Phi(\mathcal S)$.

\emph{(ii) Existence of an RPI limit set:} if $0\in\mathcal D$ and $\mathcal S_0\subseteq \Phi(\mathcal S_0)$, then
$\{\mathcal S_t\}_{t\ge0}$ is outer-monotone and converges (in the Hausdorff metric induced by $\|\cdot\|_\Omega$) to the compact convex set
\[
\mathcal E_\star
=\lim_{t\to\infty}\mathcal S_t
=\bigoplus_{i=0}^{\infty}\operatorname{co}\!\bigl(\mathcal A_K^{\,i}\mathcal D\bigr),
\]
where $\mathcal A_K^{\,0}:=\{I\}$ and
$\mathcal A_K^{\,i}:=\{A_i\cdots A_1 \mid A_\ell\in\mathcal A_K\}$ for $i\ge1$.
Moreover, $\mathcal E_\star$ satisfies the RPI inclusion
\begin{equation}\label{eq:RPI_inclusion_poly_replace}
\mathcal A_K\,\mathcal E_\star\ \oplus\ \mathcal D\ \subseteq\ \mathcal E_\star.
\end{equation}

\emph{(iii) $\varepsilon$-stopping and certified polyhedral RPI:} if, for some $t^\star$ and $\varepsilon>0$,
\begin{equation}\label{eq:stop_eps_poly_replace}
\mathcal S_{t^\star+1}\ \subseteq\ \mathcal S_{t^\star}\ \oplus\ \varepsilon\,\Omega,
\end{equation}
then the inflated set
\begin{equation}\label{eq:E_eps_poly_replace}
\mathcal E_\varepsilon
:=\mathcal S_{t^\star}\ \oplus\ \frac{\varepsilon}{1-c_\Omega}\,\Omega,
\end{equation}
is a \emph{polyhedral} convex RPI set, i.e.,
$\mathcal A_K\,\mathcal E_\varepsilon\oplus\mathcal D\subseteq\mathcal E_\varepsilon$.
\end{theorem}
\begin{proof}
(i) follows since linear images, convex hull, and Minkowski addition with fixed $\mathcal D$ preserve inclusion; compactness of
$\mathcal A_K$ (as a linear image of $\mathcal I$) and of $\mathcal D$ yields one-step boundedness.
(ii) uses $c_\Omega<1$ and the submultiplicativity of $\|\cdot\|_\Omega$ to obtain a contraction on the $\Omega$-gauge radii, yielding
uniform boundedness and convergence to $\mathcal E_\star$; the limit representation follows from repeated substitution.
(iii) By definition of $c_\Omega$, $\operatorname{co}(\mathcal A_K\Omega)\subseteq c_\Omega\Omega$.
Using linearity, convexity, and~\eqref{eq:stop_eps_poly_replace}, we obtain
\[
\begin{aligned}
\operatorname{co}(\mathcal A_K\mathcal E_\varepsilon)\oplus\mathcal D
&\subseteq
\mathcal S_{t^\star+1}
\oplus
\frac{\varepsilon}{1-c_\Omega}
\operatorname{co}(\mathcal A_K\Omega)\\
&\subseteq
\mathcal S_{t^\star}
\oplus
\left(\varepsilon+
\frac{c_\Omega\varepsilon}{1-c_\Omega}\right)\Omega
=\mathcal E_\varepsilon .
\end{aligned}
\]
Since $\mathcal A_K\mathcal E_\varepsilon\subseteq
\operatorname{co}(\mathcal A_K\mathcal E_\varepsilon)$, this proves the claimed RPI inclusion.
Finally, $\Phi$ is a $c_\Omega$-contraction in the Hausdorff metric induced by $\|\cdot\|_\Omega$;
hence, $c_\Omega<1$ ensures that~\eqref{eq:stop_eps_poly_replace} is satisfied after finitely many
iterations for every $\varepsilon>0$.
\end{proof}

Since $\mathcal A_K$ and $\mathcal D$ are polytopic, every iterate $\mathcal S_t$ generated by~\eqref{eq:Phi_iter_poly_replace}
is a polytope. Moreover, $\Omega$ is a polytope by construction, and thus whenever the $\varepsilon$-stopping
criterion~\eqref{eq:stop_eps_poly_replace} is met with a polyhedral seed $\mathcal S_0$, the certified set
$\mathcal E_\varepsilon$ in~\eqref{eq:E_eps_poly_replace} is a polyhedral RPI tube. Such sets can be computed by standard
polyhedral set-iteration algorithms; cf.~\cite{kouramas2005minimal}.

\emph{(ii) Ellipsoidal RPI sets.}
When $\mathcal I\in\{\mathcal I^{\mathcal W_{\mathrm E}},\,\mathcal I^{\mathcal V_{\mathrm E},\ast}\}$, the residual set
$\mathcal D$ in~\eqref{eq:Omega_cases} is ellipsoidal. In this case, it is natural to adopt an ellipsoidal tube of the
form $r\,\mathbb B_{P^{-1}}$, with $\mathbb B_{P^{-1}}$ induced by the same quadratic contraction certificate~\eqref{eq:common_quad_contr}.
Such ellipsoidal tubes (and their radii $r$) can be obtained directly from the contraction margin; see, e.g.,~\cite{Blanchini1999,BoydElGhaouiFeronBalakrishnan1994}.

The following theorem gives an explicit radius condition ensuring ellipsoidal robust positive invariance.
\begin{theorem}\label{thm:RPI_ellip}
Let $K$ satisfy~\eqref{eq:common_quad_contr} for the $\mathcal A_K$ induced by
$\mathcal I\in\{\mathcal I^{\mathcal W_{\mathrm E}},\,\mathcal I^{\mathcal V_{\mathrm E},\ast}\}$, and let $\mathcal D$ be as in~\eqref{eq:Omega_cases}. Define
\begin{equation}\label{eq:c_def_ellip}
c := \sup_{A_K\in\mathcal A_K}\|A_K\|_{P^{-1}}
\ \le\ \sqrt{1-\beta/\lambda_{\max}(P)}\ < 1,
\end{equation}
and $\bar d_{P^{-1}}:=\sup_{d\in\mathcal D}\|d\|_{P^{-1}}$. Then any radius $r$ satisfying
\begin{equation}\label{eq:r_choice_ellip}
r \ \ge\ \frac{\bar d_{P^{-1}}}{1-c},
\end{equation}
defines a convex RPI ellipsoid
\begin{equation}\label{eq:E_ellip}
\mathcal E_{\mathrm E}(r)
:= r\,\mathbb B_{P^{-1}}
=\{\,e\in\mathbb R^n \mid e^\top P^{-1} e \le r^2\,\},
\end{equation}
i.e., $\mathcal A_K\,\mathcal E_{\mathrm E}(r)\oplus\mathcal D\subseteq\mathcal E_{\mathrm E}(r)$.
\end{theorem}

\begin{proof}
Fix any $A_K\in\mathcal A_K$, $e\in\mathcal E_{\mathrm E}(r)$, and $d\in\mathcal D$. Then
\[
\begin{aligned}
\|e^+\|_{P^{-1}}
&=\|A_K e+d\|_{P^{-1}}
\le \|A_K\|_{P^{-1}}\,\|e\|_{P^{-1}}+\|d\|_{P^{-1}}\\
&\le c\,r+\bar d_{P^{-1}}
\le r,
\end{aligned}
\]
where the last inequality uses~\eqref{eq:r_choice_ellip}. Hence $e^+\in\mathcal E_{\mathrm E}(r)$.
\end{proof}

Moreover, for the ellipsoids in~\eqref{eq:Omega_cases}, if
$\mathcal D=\{\,d\in\mathbb R^n\mid d^\top Q_\mathcal D^{-1}d\le1\,\}$ with
$Q_\mathcal D\in\mathbb S_{\succ0}^n$, then
\[
\bar d_{P^{-1}}^2
=\sup_{d^\top Q_\mathcal D^{-1}d\le1}d^\top P^{-1}d
=\lambda_{\max}\!\bigl(Q_\mathcal D^{1/2} P^{-1} Q_\mathcal D^{1/2}\bigr).
\]
In particular, for $\mathcal I^{\mathcal V_{\mathrm E},\ast}$ we have $Q_\mathcal D=(1+\gamma_{\mathrm E}^\ast)^2Q_v$, and for
$\mathcal I^{\mathcal W_{\mathrm E}}$ we have $Q_\mathcal D=Q_w$. Consequently, Theorem~\ref{thm:RPI_ellip} provides a certified ellipsoidal tube that can be used to define the terminal and initial tube sets required for recursive feasibility.

\begin{remark}\label{rem:poly_vs_ellip}
The limit set $\mathcal E_\star$ in Theorem~\ref{thm:RPI_poly} coincides with the 
mRPI set of the linear difference inclusion~\cite{kouramas2005minimal} associated with~\eqref{eq:err_dyn_unified}, and is therefore the smallest RPI set in the sense of set inclusion. Under the same noise scaling in~\eqref{eq:Omega_cases}, the polyhedral iteration~\eqref{eq:Phi_def_poly_replace} typically yields tighter tubes than the ellipsoid $r\mathbb B_{P^{-1}}$ in Theorem~\ref{thm:RPI_ellip}, potentially enlarging the feasible terminal set.

This improved tightness comes at a computational cost: the convex-hull operation in $\Phi(\cdot)$ and the vertex-based LMI conditions scale with the state dimension and the number of vertices, which can grow rapidly. The ellipsoidal construction avoids set iteration and vertex enumeration and is therefore often easier to obtain in large-scale settings.
\end{remark}
\begin{remark}\label{rem:gamma-star}
Relative to the known-model case, conservatism arises from outer-approximating the exact measurement-residual set $\mathcal V\oplus(-A^\ast\mathcal V)$ by $(1+\gamma^\ast)\mathcal V$ and replacing $(A^\ast,B^\ast)$ by its data-consistency set. Moreover, the samplewise enclosure~\eqref{eq:D_gamma_general} does not exploit cross-sample consistency induced by measurement-noise variables shared by adjacent residuals; for example, $v_{k+1}$ enters both $\eta_k$ and $\eta_{k+1}$.

More informative data, obtained for example from longer trajectories
and richer, persistently exciting inputs, typically contract the
data-consistency sets defined by the accumulated finite constraints.
Since the refinement maps in Appendix~\ref{app:gamma-tight} maximize
the induced gain over these sets, their contraction in turn typically
yields tighter certified bounds $\gamma_{\mathrm P}^\ast$ and
$\gamma_{\mathrm E}^\ast$ without sacrificing containment of
$(A^\ast,B^\ast)$. Section~\ref{sec:examples} quantifies the resulting
conservatism.
\end{remark}
\section{Example}\label{sec:examples}
We illustrate the constructions on the flight-vehicle example sampled every $1\,\mathrm{s}$, as in~\cite{Mayne2005},
under the measurement-noise model, with
\[
\begin{aligned}
x_{k+1}&=\begin{bmatrix}1&1\\0&1\end{bmatrix}x_k
+\begin{bmatrix}0.5\\1\end{bmatrix}u_k,
&\hat x_k&=x_k+v_k,\\
\mathcal V_{\mathrm P}(\bar v)&:=\{\,v\in\mathbb R^2\mid\|v\|_\infty\le\bar v\,\},\\
\mathcal V_{\mathrm E}(\bar v)&:=\{\,v\in\mathbb R^2\mid v^\top Q_v^{-1}v\le1\,\},
&Q_v&:=\bar v^2I.
\end{aligned}
\]
We sweep both the data length and the noise level, using five logarithmically spaced values of $T$ over $[10,1000]$ and five values of $\bar v$ over $[0.01,0.09]$. For each $T$, the input is drawn i.i.d.\ with $u_k^d\in[-3,3]$; across the five noise bounds, the same input and normalized measurement-noise realizations are used, with the latter scaled by $\bar v$. The certified bounds $\gamma_{\mathrm P}^\ast$ and $\gamma_{\mathrm E}^\ast$ are computed as described in
Appendix~\ref{app:gamma-tight}. We then compute robust stabilizing gains
by solving the SDPs in Propositions~\ref{prop:K-stab-ellip} and~\ref{prop:K-stab-poly}.

Finally, for each $(T,\bar v)$ we construct the corresponding polyhedral and ellipsoidal RPI tube sets.
For either geometry, the reported tube-size gap is the relative area increase
$100[\operatorname{area}(\mathcal E_{\rm data})-\operatorname{area}(\mathcal E_{\rm model})]/
\operatorname{area}(\mathcal E_{\rm model})$, where the model-based counterpart is computed using
$(A^\ast,B^\ast)$ and the same feedback gain. Fig.~\ref{fig:tube_gap} summarizes the resulting percentage gaps.
Overall, the observed gap generally decreases as $T$ increases and $\bar v$ decreases. Consistent with Remark~\ref{rem:gamma-star}, longer trajectories typically provide more informative finite constraints, while smaller noise bounds reduce residual inflation; together, these effects can yield tighter data-consistency sets and certified induced-gain bounds. At $T=1000$ and $\bar v=0.01$, the gap is
about $1.0\%$ for the polyhedral tube and about $5.2\%$ for the ellipsoidal tube, indicating that the induced
conservatism is numerically mild.

\begin{figure}[t]
  \centering
  \subfloat[Polyhedral tube.\label{fig:tube_gap_poly}]{\includegraphics[width=1\linewidth]{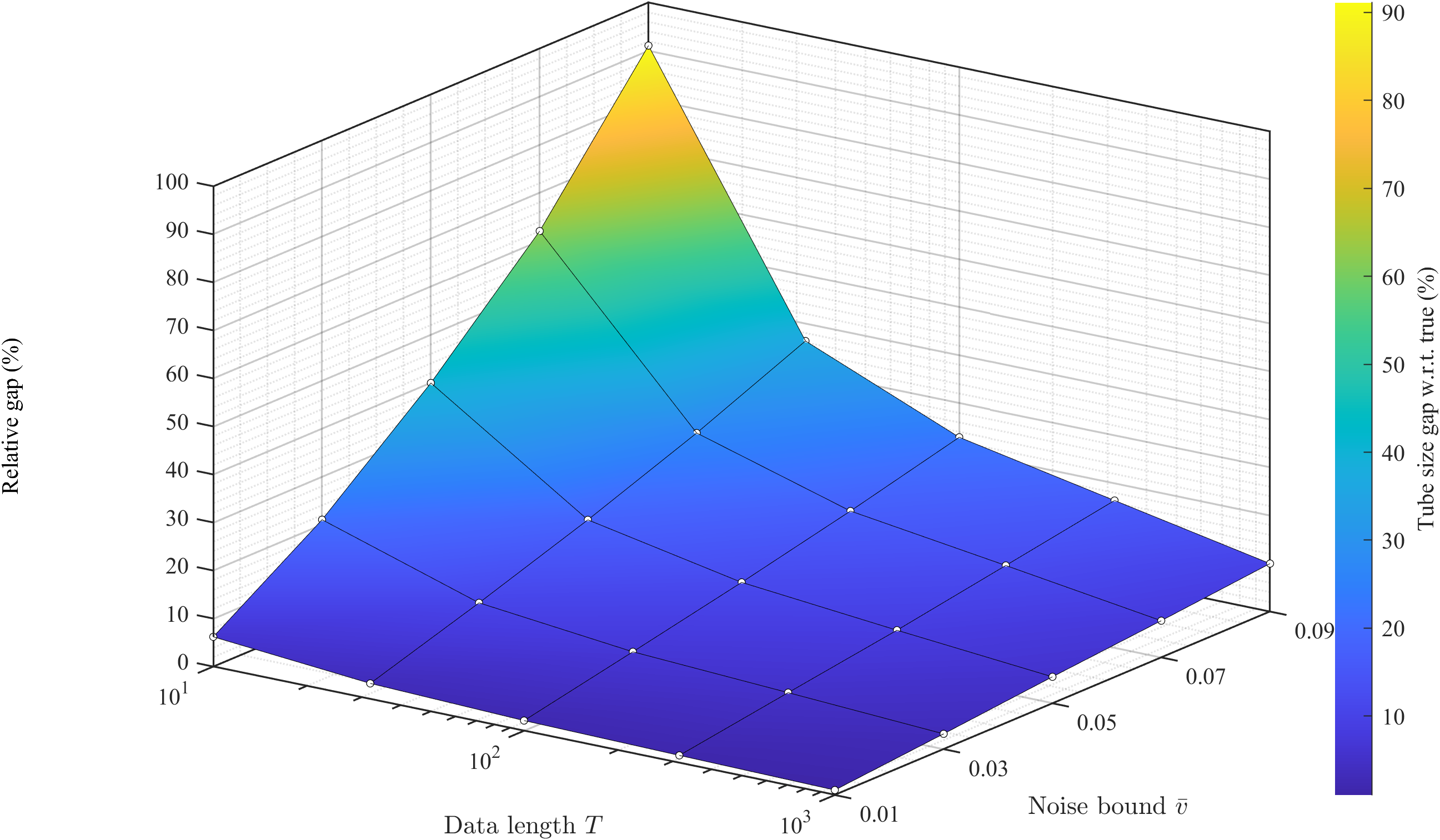}}\hfill
  \subfloat[Ellipsoidal tube.\label{fig:tube_gap_ellip}]{\includegraphics[width=1\linewidth]{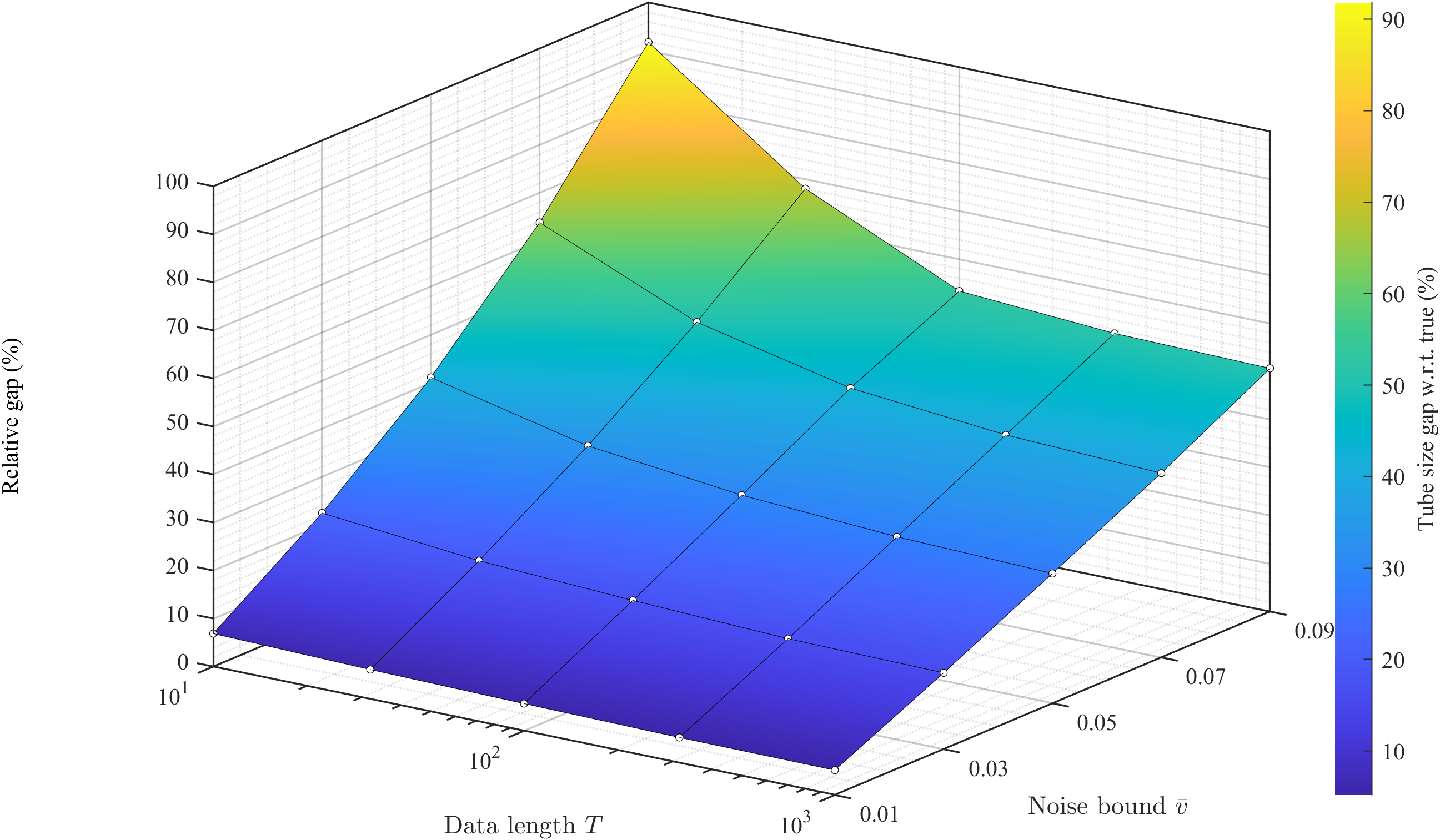}}
  \caption{Relative tube-size gap (percentage) w.r.t.\ the model-based tube computed using $(A^\ast,B^\ast)$,
  over a grid of $(T,\bar v)$.}\label{fig:tube_gap}
\end{figure}

To highlight the conservatism introduced by the lossy matrix S-procedure in the synthesis step, we consider the
scalar measurement-noise model
\[
x_{k+1}=1.1x_k+0.6u_k,\qquad \hat x_k=x_k+v_k.
\]
In one dimension, the polytopic and ellipsoidal noise descriptions coincide; consequently, for a common
choice of $\gamma^\ast$ the induced data-consistency sets also coincide. Therefore, any observed gap between the gains and contraction margins certified by
Proposition~\ref{prop:K-stab-ellip} and Proposition~\ref{prop:K-stab-poly} is attributable solely to the
conservatism of the underlying synthesis certificate.

We fix one offline trajectory of length $T=100$ and sweep the noise bound~$\bar v$ over $[0.1,0.2]$, where
the difference between the two certificates is more visible. For each $\bar v$, we
synthesize a robust feedback gain by maximizing the certified contraction margin $\beta$
under the same normalization (here $P=1$). Let $c:=\sup_{A_K\in\mathcal A_K}|A_K|$,
$d:=(1+\gamma^\ast)\bar v$, and $s_\infty:=d/(1-c)$, so that the mRPI interval has length $2s_\infty$.
As shown in Fig.~\ref{fig:1d_sproc_vs_vtx}, the
lossy S-procedure certifies a smaller margin and yields a larger mRPI tube, and the discrepancy increases
with~$\bar v$.

\begin{figure}[t]
  \centering
  \subfloat[Certified margin $\beta$.\label{fig:1d_beta}]{\includegraphics[width=0.50\linewidth]{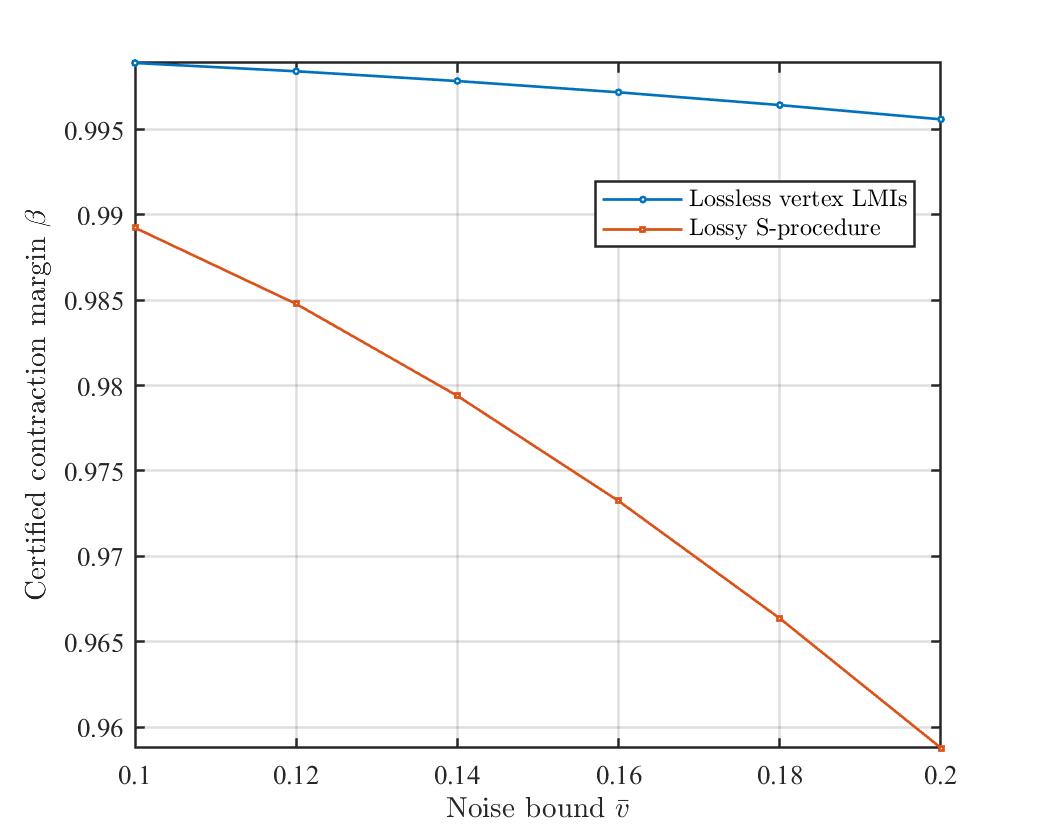}}\hfill
  \subfloat[mRPI tube size $2s_\infty$.\label{fig:1d_rpi}]{\includegraphics[width=0.50\linewidth]{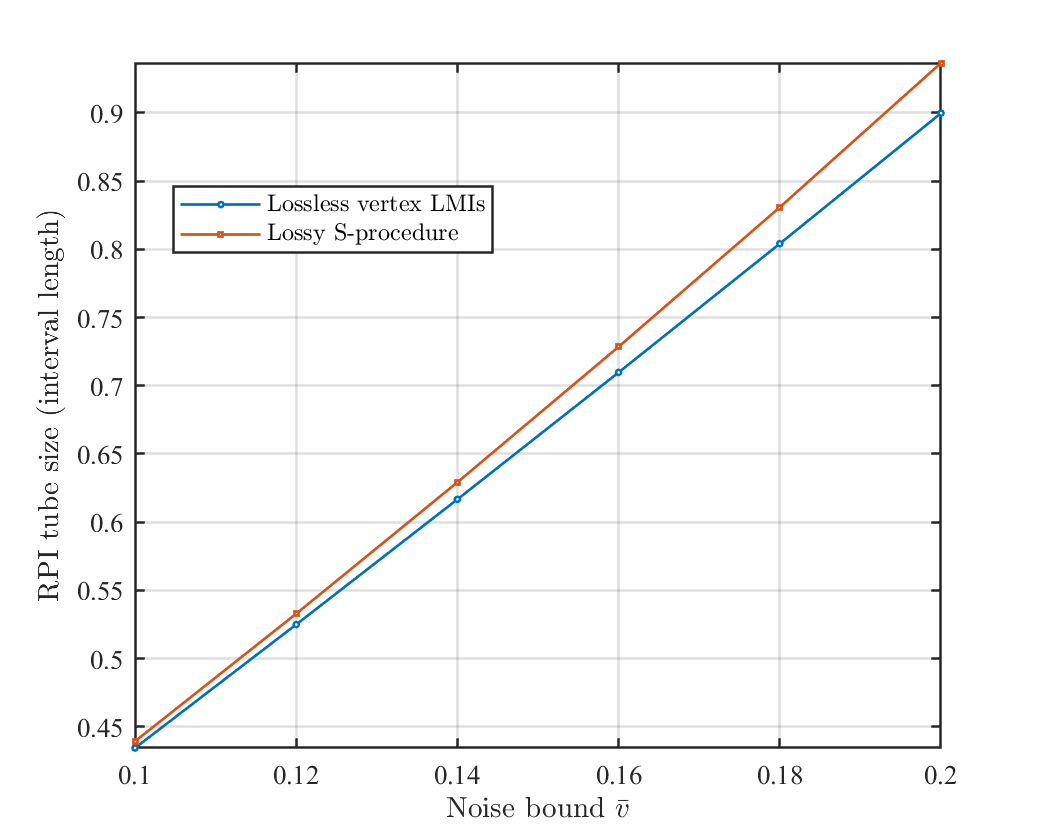}}
  \caption{Lossless vertex LMIs vs.\ lossy S-procedure under the same
  data.}\label{fig:1d_sproc_vs_vtx}
\end{figure}

\section{Conclusion}\label{sec:conclusions}
This paper provided a constructive, offline-computable approach to obtain certified RPI tubes from finite noisy input--state data of an unknown LTI system, directly supporting tube-based robust data-driven predictive control. The analysis identifies residual inflation and data-consistency uncertainty as key sources of conservatism: polyhedral iteration typically yields tighter tubes, whereas ellipsoidal certification avoids vertex enumeration and set iteration. Numerical results quantify the induced conservatism and show that it decreases as the data become more informative.

Several topics for future research remain open. The extension of the proposed data-driven RPI synthesis to tube-based robust data-driven predictive control with closed-loop guarantees has recently been established in related work~\cite{WangAngeli2026TRDDPC}. Another challenging direction is to extend the framework to nonlinear systems and to incorporate nonlinear feedback laws to further reduce conservatism.

\appendices
\section{Data-consistent induced gain bounds}\label{app:gamma-tight}
For the measurement-noise data, let
$J_A:=\bigl[\begin{smallmatrix}I_n\\0_{m\times n}\end{smallmatrix}\bigr]$,
$L:=Z^\dagger J_A$, and $\hat A:=\hat X_1L$.
For a given measurement-noise set $\mathcal V$, let $\alpha_{\mathcal V}$ be any computable constant satisfying
\begin{equation}\label{eq:alpha_V_bound_app}
\max_{\xi_i\in\mathcal V,\ i=0,\ldots,T-2}
\|[\xi_0\ \cdots\ \xi_{T-2}]L\|_{\mathcal V}
\le\alpha_{\mathcal V}.
\end{equation}

\begin{lemma}\label{lem:gamma_initial_certificate}
Suppose Assumption~\ref{as:fullrowrank} holds and $\alpha_{\mathcal V}<1$. Then
\begin{equation}\label{eq:gamma_V_initial_app}
\bar\gamma_{\mathcal V}
:=
\frac{\mu_{\mathcal V}(\hat A)+\alpha_{\mathcal V}}
{1-\alpha_{\mathcal V}}
\end{equation}
satisfies $\mu_{\mathcal V}(A^\ast)\le\bar\gamma_{\mathcal V}$ and
$(A^\ast,B^\ast)\in\mathcal I^{\mathcal V}(\bar\gamma_{\mathcal V})$.
Moreover, every $(A,B)\in\mathcal I^{\mathcal V}(\gamma)$ satisfies
\begin{equation}\label{eq:gain_envelope_app}
\mu_{\mathcal V}(A)
\le\mu_{\mathcal V}(\hat A)+(1+\gamma)\alpha_{\mathcal V}.
\end{equation}
\end{lemma}
\begin{proof}
Let $q^\ast:=\mu_{\mathcal V}(A^\ast)$ and
$E^\ast:=[\eta_0^\ast\ \cdots\ \eta_{T-2}^\ast]$. Lemma~\ref{lem:gauge_inflation_meas} gives
$\eta_i^\ast/(1+q^\ast)\in\mathcal V$, and hence~\eqref{eq:alpha_V_bound_app} gives
$\|E^\ast L\|_{\mathcal V}\le(1+q^\ast)\alpha_{\mathcal V}$. Since
$\hat X_1=[A^\ast\ B^\ast]Z+E^\ast$ and $ZZ^\dagger=I_{n+m}$, we have
$-A^\ast=-\hat A+E^\ast L$.
The triangle inequality for the induced gain therefore yields
$q^\ast\le\mu_{\mathcal V}(\hat A)+(1+q^\ast)\alpha_{\mathcal V}$, which implies
$q^\ast\le\bar\gamma_{\mathcal V}$. Lemma~\ref{lem:gauge_inflation_meas} then gives
$(A^\ast,B^\ast)\in\mathcal I^{\mathcal V}(\bar\gamma_{\mathcal V})$.

For any $(A,B)\in\mathcal I^{\mathcal V}(\gamma)$, let
$E$ collect its fitting residuals. Then the columns of $E/(1+\gamma)$ belong to
$\mathcal V$ and $-A=-\hat A+EL$. Applying~\eqref{eq:alpha_V_bound_app} to
$E/(1+\gamma)$ gives~\eqref{eq:gain_envelope_app}.
\end{proof}
If $\alpha_{\mathcal V}\ge1$, this sufficient certificate is inconclusive and additional data or a prior gain bound is required.

\emph{(i) Polytopic measurement noise.}
For $\mathcal V=\mathcal V_{\mathrm P}$ and any $\gamma$ for which
$\mathcal I^{\mathcal V_{\mathrm P}}(\gamma)$ is nonempty, define
\[
\begin{aligned}
f_{\mathrm P}(\gamma)
&:=\max_{(A,B)\in\mathcal I^{\mathcal V_{\mathrm P}}(\gamma)}
\mu_{\mathcal V_{\mathrm P}}(A),\\
\mu_{\mathcal V_{\mathrm P}}(A)
&=\max_{\ell}\max_{v\in\operatorname{vert}(\mathcal V_{\mathrm P})}
\frac{-g_{v,\ell}^\top Av}{(h_v)_\ell},
\end{aligned}
\]
where $g_{v,\ell}^\top$ denotes the $\ell$th row of $G_v$.
Consequently, both $\alpha_{\mathcal V_{\mathrm P}}$ and $f_{\mathrm P}(\gamma)$ reduce to finitely many linear-program evaluations. Initialize
$\gamma_{\mathrm P,0}:=\bar\gamma_{\mathcal V_{\mathrm P}}$ and iterate
\begin{equation}\label{eq:gammaP_refinement_app}
\gamma_{\mathrm P,r+1}:=f_{\mathrm P}(\gamma_{\mathrm P,r}).
\end{equation}
By~\eqref{eq:gamma_V_initial_app} and~\eqref{eq:gain_envelope_app},
$f_{\mathrm P}(\gamma_{\mathrm P,0})\le\gamma_{\mathrm P,0}$, and $f_{\mathrm P}$ is
nondecreasing because $\mathcal I^{\mathcal V_{\mathrm P}}(\gamma)$ is nested in $\gamma$.
Thus $\gamma_{\mathrm P,r+1}\le\gamma_{\mathrm P,r}$ implies
$\gamma_{\mathrm P,r+2}=f_{\mathrm P}(\gamma_{\mathrm P,r+1})
\le f_{\mathrm P}(\gamma_{\mathrm P,r})=\gamma_{\mathrm P,r+1}$.
Moreover, if $(A^\ast,B^\ast)\in\mathcal I^{\mathcal V_{\mathrm P}}(\gamma_{\mathrm P,r})$, then
$\mu_{\mathcal V_{\mathrm P}}(A^\ast)\le\gamma_{\mathrm P,r+1}$, and
Lemma~\ref{lem:gauge_inflation_meas} certifies membership at $\gamma_{\mathrm P,r+1}$.
Starting from Lemma~\ref{lem:gamma_initial_certificate}, induction gives
$\mu_{\mathcal V_{\mathrm P}}(A^\ast)\le\gamma_{\mathrm P,r+1}\le\gamma_{\mathrm P,r}$;
hence every iterate is certified. Since this nonincreasing sequence is bounded below, its successive differences
converge to zero. For a prescribed tolerance $\varepsilon_\gamma>0$, choose $r^\ast$ such that
$\gamma_{\mathrm P,r^\ast}-\gamma_{\mathrm P,r^\ast+1}\le\varepsilon_\gamma$ and set
$\gamma_{\mathrm P}^\ast:=\gamma_{\mathrm P,r^\ast+1}$.

\emph{(ii) Ellipsoidal measurement noise.}
For $\mathcal V=\mathcal V_{\mathrm E}$, the computable choice
$\alpha_{\mathrm E}:=\sqrt{T-1}\,\|LQ_v^{1/2}\|_2$ satisfies~\eqref{eq:alpha_V_bound_app}: writing
$[\xi_0\ \cdots\ \xi_{T-2}]=Q_v^{1/2}W$ gives $\|W\|_2\le\sqrt{T-1}$ and
$\|[\xi_0\ \cdots\ \xi_{T-2}]L\|_{\mathcal V_{\mathrm E}}=\|WLQ_v^{1/2}\|_2
\le\sqrt{T-1}\|LQ_v^{1/2}\|_2$.
Since $\mathcal V_{\mathrm E}$ is centrally symmetric, if $\alpha_{\mathrm E}<1$ we set
\begin{equation}\label{eq:gammaE_star_app_tight}
\gamma_{\mathrm E}^\ast
:=
\frac{\|Q_v^{-1/2}\hat A Q_v^{1/2}\|_2+\alpha_{\mathrm E}}
{1-\alpha_{\mathrm E}}.
\end{equation}
Lemma~\ref{lem:gamma_initial_certificate} then certifies~\eqref{eq:gamma_star_req} and
$(A^\ast,B^\ast)\in\mathcal I^{\mathcal V_{\mathrm E}}(\gamma_{\mathrm E}^\ast)$.
An analogous certified, nonincreasing refinement is
$\gamma_{\mathrm E,r+1}:=f_{\mathrm E}(\gamma_{\mathrm E,r})$, initialized at
$\gamma_{\mathrm E,0}:=\gamma_{\mathrm E}^\ast$, where
$f_{\mathrm E}(\gamma):=\max_{(A,B)\in\mathcal I^{\mathcal V_{\mathrm E}}(\gamma)}
\|A\|_{\mathcal V_{\mathrm E}}$. Since evaluating $f_{\mathrm E}$ is a nonconvex norm-maximization problem,
we use~\eqref{eq:gammaE_star_app_tight} directly and make no minimality claim.
\end{document}